\documentclass{llncs}
\usepackage{graphicx,wrapfig}

\newcommand{\GI}{G_{DC}}
\newcommand{\oli}{\overline}
\newcommand{\ox}{\overline{x}}
\newtheorem{cor}{Corollary}

\graphicspath{{figures/}}

\begin{document}

\mainmatter

\title{Edges and Switches, Tunnels and Bridges}

\author{D. Eppstein\inst{1} \and M. van Kreveld\inst{2} \and E. Mumford\inst{3} \and B. Speckmann\inst{3}}

%\titlerunning{Edges and Switches, Tunnels and Bridges}
%\authorrunning{David Eppstein et al.}
%
%\tocauthor{David Eppstein (University of California, Irvine), Marc van Kreveld (Utrecht
%University), Elena Mumford (TU Eindhoven), Bettina Speckmann (TU Eindhoven)}

\institute{Department of Computer Science, University of California, Irvine, \email
        {eppstein@ics.uci.edu}
\and Department of Information and Computing Sciences, Utrecht
University, \email{marc@cs.uu.nl} \and Department of
Mathematics and Computer Science, TU Eindhoven,
\email{e.mumford@tue.nl} and \email{speckman@win.tue.nl}}

%---------------------- Text ----------------------------------------------

\maketitle

\begin{abstract}\noindent
Edge casing is a well-known method to improve the readability
of drawings of non-planar graphs. A cased drawing orders the
edges of each edge crossing and interrupts the lower edge in an
appropriate neighborhood of the crossing. Certain orders will
lead to a more readable drawing than others. We formulate
several optimization criteria that try to capture the concept
of a ``good'' cased drawing. Further, we address the
algorithmic question of how to turn a given drawing into an
optimal cased drawing. For many of the resulting optimization
problems, we either find polynomial time algorithms or
NP-hardness results.
\end{abstract}

\section{Introduction}

Drawings of non-planar graphs necessarily contain edge
crossings. The vertices of a drawing are commonly marked with a
disk, but it can still be difficult to detect a vertex within a
dense cluster of edge crossings. \emph{Edge casing} is a
well-known method---used, for example, in electrical drawings
and, more generally, in information visualization---to
alleviate this problem and to improve the readability of a
drawing. A \emph{cased drawing} orders the edges of each
crossing and interrupts the lower edge in an appropriate
neighborhood of the crossing. One can also envision that every
edge is encased in a strip of the background color and that the
casing of the upper edge covers the lower edge at the crossing.
See Fig.~\ref{fig:crossingdrawing} for an example.
\begin{figure}[b]
\centering
\includegraphics{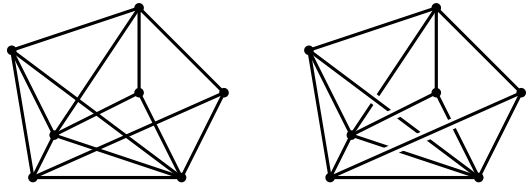}
\caption{Normal and cased drawing of a graph.}
\label{fig:crossingdrawing}
\end{figure}

If there are no application specific restrictions that dictate
the order of the edges at each crossing, then we can in
principle choose freely how to arrange them. Certain orders
will lead to a more readable drawing than others. In this paper
we formulate several optimization criteria that try to capture
the concept of a ``good'' cased drawing. Further, we address
the algorithmic question of how to turn a given drawing into an
optimal cased drawing.

\paragraph{\bfseries Definitions.} Let $G$ be a graph with $n$ vertices and $m$ edges and let $D$ be a drawing of $G$ with $k$ crossings. We want to turn $D$ into a cased drawing where the width of the casing is given in the variable $\mathit{casing width}$. To avoid that the casing of an edge covers a vertex we assume that no vertex $v$ of $D$ lies on (or very close to) an edge $e$ of $D$ unless $v$ is an endpoint of $e$. Further, no more than two edges of $D$ cross in one point and any two crossings are far enough apart so that the casings of the edges involved do not interfere. With these assumptions we can consider crossings independently. Without these restrictions the problem changes significantly---optimization problems that are solvable in polynomial time can become NP-hard. Additional details can be found in Appendix~\ref{sec:restrictions}.

We define the \emph{edge crossing graph} $\GI$ for $D$ as follows. $\GI$ contains a vertex for every edge of $D$ and an edge for any two edges of $D$ that cross.
Let $C$ be a crossing between two edges $e_1$ and $e_2$. In a
cased drawing either $e_1$ is drawn on top of $e_2$ or vice
versa. If $e_1$ is drawn on top of $e_2$ then we say that $C$
is a \emph{bridge} for $e_1$ and a \emph{tunnel} for $e_2$. In
Fig.~\ref{fig:tunnelbridge}, $C_1$ is a bridge for $e_1$ and a
tunnel for $e_2$. The \emph{length} of a tunnel is ${\mathit{casing width}}/\sin\alpha$, where $\alpha\leq \pi/2$ is the angle of the edges at the crossing.
A pair of consecutive crossings $C_1$ and
$C_2$ along an edge $e$ is called a \emph{switch} if $C_1$ is a
bridge for $e$ and $C_2$ is a tunnel for $e$, or vice versa. In
Fig.~\ref{fig:tunnelbridge}, $(C_1, C_2)$ is a switch.
\begin{figure}[h]
  \centering
  \vspace{-\baselineskip}
  \includegraphics{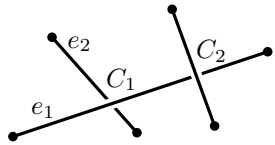}
  \vspace{-.25\baselineskip}
  \caption{Tunnels and bridges.}
  \vspace{-2\baselineskip}
  \label{fig:tunnelbridge}
\end{figure}

\paragraph{\bfseries Stacking and weaving.} When we turn a given drawing into a cased drawing, we need
to define a drawing order for every edge crossing. We can
choose to either establish a global top-to-bottom order on the
edges, or to treat each edge crossing individually. We call the
first option the \emph{stacking model} and the second one the
\emph{weaving model}, since cyclic overlap of three or more
edges can occur (see Fig.~\ref{fig:stack-weave}).
\begin{figure}[h]
  \centering
  \vspace{-\baselineskip}
  \includegraphics{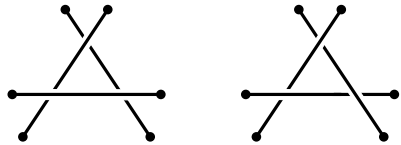}
  \vspace{-.25\baselineskip}
  \caption{Stacking and weaving.}
  \label{fig:stack-weave}
\end{figure}

\paragraph{\bfseries Quality of a drawing.}
Globally speaking, two factors may influence the readability of
a cased drawing in a negative way. Firstly, if there are many
switches along an edge then it might become difficult to follow
that edge. Drawings that have many switches can appear somewhat
chaotic. Secondly, if an edge is frequently below other edges,
then it might become hardly visible. These two considerations
lead to the following optimization problems for a drawing $D$.
\begin{description}
\item[{\sc MinTotalSwitches}] Minimize the total number of switches.
\item[{\sc MinMaxSwitches}] Minimize the maximum number of switches for any edge.
\item[{\sc MinMaxTunnels}]  Minimize the maximum number of tunnels for any edge.
\item[{\sc MinMaxTunnelLength}] Minimize the maximum total length of tunnels for any edge.
\item[{\sc MaxMinTunnelDistance}] Maximize the minimum distance between any two consecutive tunnels.
\end{description}
Fig.~\ref{fig:weavingopt} illustrates that the weaving model is
stronger than the stacking model for {\sc
MinTotalSwitches}---no cased drawing of this graph in the
stacking model can reach the optimum of four switches. For, the
thickly drawn bundles of $c>4$ parallel edges must be cased as
shown (or its mirror image) else there would be at least $c$
switches in a bundle, the four vertical and horizontal segments
must cross the bundles consistently with the casing of the
bundles, and this already leads to the four switches that occur
as drawn near the midpoint of each vertical or horizontal
segment. Thus, any deviation from the drawing in the casing of
the four crossings between vertical and horizontal segments
would create additional switches. However, the drawing shown is
not a stacked drawing.

\begin{figure}[htb]
  \centering
  \vspace{-.5\baselineskip}
  \includegraphics[width=.6\columnwidth]{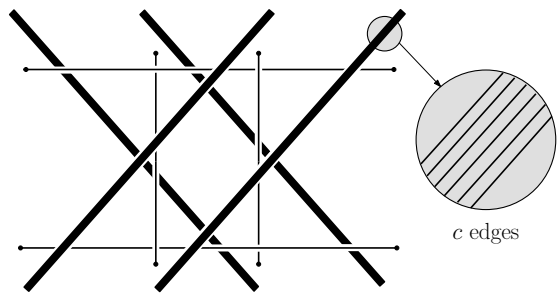}
  \caption{Optimal drawing in the weaving model for {\sc MinTotalSwitches}.}
  \label{fig:weavingopt}
  \vspace{-1.5\baselineskip}
\end{figure}

\paragraph{\bfseries Related work.} If we consider only simple arrangements of line segments in the plane as our initial drawing, then there is a third model to consider, an intermediate between stacking and weaving: drawings which are plane projections of line segments in three dimensions. We call this model the \emph{realizable model}. Clearly every cased drawing in the stacking model is also a drawing in the realizable model, but not every cased drawing in the weaving model can be realized (see~\cite{PacPolWel-Algo-93}). The optimal drawing in Fig.~\ref{fig:weavingopt} can be realized, hence the realizable model is stronger than the stacking model. In Appendix~\ref{sec:optweave} we show that the weaving model is stronger than the realizable model.

\paragraph{\bfseries Results.} For many of the problems described above, we either find
polynomial time algorithms or NP-hardness results in both the
stacking and weaving models. We summarize our results in Table~\ref{tbl:results}. In this paper we assume that our input drawing is a straight line drawing, but several of our results also generalize to curved drawings. Section~\ref{sec:switches} presents the results concerning the optimization problems that seek to minimize the number of switches and Section~\ref{sec:tunnels} discusses our solutions to the optimization problems that concern the tunnels. In Appendix~\ref{sec:restrictions} we show that {\sc MinTotalSwitches} becomes NP-hard in both the weaving and the stacking model if we allow more than three edges to cross in one point. We conclude with some open problems.

\begin{table}[t]
\centering
\begin{tabular}{|l|l|l|}
\hline
Model & Stacking & Weaving \\
\hline
{\sc MinTotalSwitches} & \emph{open} & $O(qk+q^{5/2}\log^{3/2} k)$\\
{\sc MinMaxSwitches} & \emph{open} & \emph{open} \\
{\sc MinMaxTunnels} & $O(m\log m+k)$ \emph{exp.} & $O(m^4)$ \\
{\sc MinMaxTunnelLength} & $O(m\log m+k)$ \emph{exp.} & NP-hard \\
{\sc MaxMinTunnelDistance} & $O(m\log m+k\log m)$ \emph{exp.} & $O((m+K)\log m)$ \emph{exp.}\\
\hline
\end{tabular}
\medskip
\caption{Table of results: $n$ is the number of vertices,
$m=\Omega(n)$ is the number of edges, $K=O(m^3)$ is
the total number of pairs of crossings on the same edge, $k = O(m^2)$ is the
number of crossings of the input drawing, and $q = O(k)$ is
the number of its \emph{odd face polygons}.}
\label{tbl:results} \vspace{-\baselineskip}
\end{table}

\section{Minimizing switches}\label{sec:switches}

In this section we discuss results related to the {\sc
MinTotalSwitches} and {\sc MinMaxSwitches} problems. We first
discuss some non-algorithmic results giving simple bounds on
the number of switches needed, and recognition algorithms for
graphs needing no switches. As we know little about these
problems for the stacking model, all results stated in this
section will be for the weaving model.

\begin{lemma} Given a drawing $D$ of a graph we can turn $D$ into a cased
drawing without any switches if and only if the edge crossing
graph $\GI$ is bipartite.
\end{lemma}
\begin{cor}Given a drawing $D$ of a graph we can decide in
$O((n+m)\log(n+m))$ time if $D$ can be turned into a cased
drawing without any switches.
\end{cor}
\begin{proof}
We apply the bipartiteness algorithm of \cite{Epp-SODA-04}.
Note that this does not construct the arrangement, so there is
no term with $k$ in the runtime.\qed
\end{proof}
Define a \emph{vertex-free cycle} in a drawing of a graph $G$
to be a face $f$ formed by the arrangement of the edges in the
drawing, such that there are no vertices of $G$ on the boundary
of $f$. An \emph{odd vertex-free cycle} is a vertex-free cycle
composed of an odd number of segments of the arrangement.
\begin{lemma}
\label{lem:odd-cycle-has-one-switch} Let $f$ be an odd
vertex-free cycle in a drawing $D$. Then in any casing of $D$,
there must be a switch on one of the segments of $f$.
\end{lemma}
\begin{proof}
Unless there is a switch, the segments must alternate between
those that cross above the previous segment, and those that
cross below the previous segment. However, this alternation
cannot continue all the way around an odd cycle, for it would
end up in an inconsistent state from how it started.\qed
\end{proof}
\begin{lemma}
Given a drawing $D$ of a graph the minimum number of switches
of any cased drawing obtained from $D$ is at least half of the
number of odd vertex-free cycles in $D$.
\end{lemma}
\begin{proof}
Let $o$ be the number of odd vertex-free cycles in $D$. By
Lemma~\ref{lem:odd-cycle-has-one-switch}, each odd vertex-free
cycle must have a switch on one of its segments. Choose one
such switch for each cycle; then each segment belongs to at
most two vertex-free cycles, so these choices group the odd
cycles into pairs of cycles sharing a common switch, together
with possibly some unpaired cycles. The number of pairs and
unpaired cycles must be at least $o/2$, so the number of
switches must also be this large.
\end{proof}
\vspace{-1.5\baselineskip}

\begin{wrapfigure}[9]{r}{.38\textwidth}
  \centering
  \vspace{-\baselineskip}
  \includegraphics{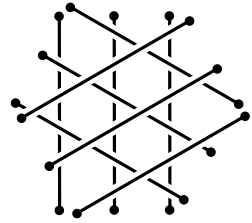}
  \vspace{-.5\baselineskip}
  \small
  \caption{A construction with $O(n)$ edges and $\Omega(n^2)$ triangles.}
  \label{fig:n2triangles}
\end{wrapfigure}
\hspace{2in}
\begin{lemma}
For any $n$ large enough, a drawing of a graph $G$ with $n$
vertices and $O(n)$ edges exists for which any crossing choice
gives rise to $\Omega(n^2)$ switches.
\end{lemma}
\begin{proof} A construction with three sets of parallel
lines, each of linear size, gives $\Omega(n^2)$ vertex-free
triangles, and each triangle gives at least one switch (see
Fig.~\ref{fig:n2triangles}).\qed
\end{proof}

\begin{lemma}
For any $n$ large enough, a drawing of a graph $G$ with $n$
vertices and $O(n^2)$ edges exists for which any crossing
choice gives rise to $\Omega(n^4)$ switches.
\end{lemma}
\begin{proof}
We build our graph as follows: make a very elongated
rectangle, place $n/6$ vertices equally spaced on each short
edge, and draw the complete bipartite graph. This graph has
$(n/6)^2$ edges. One can prove that there is a strip parallel
to the short side of the rectangle, such that the parts of the
edges inside the strip behave in the same way as parallel ones
do with respect to creating triangles when overlapped the way
it is described in the previous lemma. This gives us the
desired graph with $\Omega(n^4)$ triangles, and hence with
$\Omega(n^4)$ switches.\qed
\end{proof}

\bigskip\noindent
We define a \emph{degree-one graph} to be a graph in which
every vertex is incident to exactly one edge; that is, it must
consist of a collection of disconnected edges.

\begin{lemma}
\label{lem:degree-one-equivalent} Let $D$ be a drawing of a
graph $G$. Then there exists a drawing $D'$ of a degree-one
graph $G'$, such that the edges of $D$ correspond one-for-one
with the edges of $D'$, casings of $D$ correspond one-for-one
to casings of $D'$, and switches of $D$ correspond one-for-one
with switches of $D'$.
\end{lemma}
\begin{proof}
Form $G'$ by placing a small circle around each vertex of $G$.
Given an edge $e=(u,v)$ in $G$, let $u_e$
be the point where $e$ crosses the circle around $u$ and similarly let $v_e$ be the
point where $e$ crosses the circle around $v$. Form $D'$ and
$G'$ by replacing each edge $e=(u,v)$ in $G$ by the corresponding
edge $(u_e,v_e)$, drawn as the subset of edge $e$ connecting
those points.

As these replacements do not occur between any two crossings
along any edge, they do not affect the switches on the edge.
Both drawings have the same set of crossings, and any switch in
a casing of one drawing gives rise to a switch in the
corresponding casing of the other drawing.\qed
\end{proof}
In a drawing of a degree-one graph, define a \emph{polygon} to
be a sequence of segments of the arrangement formed by the
drawing edges that forms the boundary of a simple polygon in
the plane. Define a \emph{face polygon} to be a polygon that
forms the boundary of the closure of a face of the arrangement;
note that there may be edges drawn in the interior of this
polygon, as long as they do not separate it into multiple
components.
\begin{lemma}
In a drawing of a degree-one graph, there can be no vertex on
any segment of a polygon.
\end{lemma}
\begin{proof}
We have already required that no vertex can lie on an edge
unless it is the endpoint of an edge. And, if a segment
contains the endpoint of an edge, it cannot continue past the
endpoint to form the boundary of a polygon.\qed
\end{proof}

\begin{wrapfigure}[8]{r}{.45\textwidth}
  \centering
  \vspace{-1.75\baselineskip}
  \includegraphics{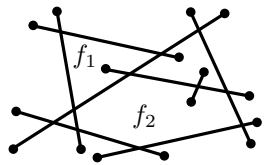}
  \vspace{-.5\baselineskip}
  \small
  \caption{A degree-one graph, $f_1$ is an odd polygon and $f_2$ is an even polygon.}
  \label{fig:polygon}
\end{wrapfigure}
\noindent
Note, however, that a polygon can contain vertices in its
interior. Define the \emph{complexity} of a polygon to be the
number of segments forming it, plus the number of graph
vertices interior to the polygon. We say that a polygon is
\emph{odd} if its complexity is an odd number, and \emph{even}
if its complexity is an even number (see Fig.~\ref{fig:polygon}).

\begin{lemma}
\label{lem:polygon-is-sum-of-faces} Let $p$ be a polygon in a
drawing of a degree-one graph. Then, modulo two, the complexity
of $p$ is equal to the sum of the complexities of the face
polygons of faces within $p$.
\end{lemma}
\begin{proof}
Each segment of $p$ contributes one to the complexity of $p$
and one to the complexity of some face polygon. Each vertex
within $p$ contributes one to the complexity of $p$ and one to
the complexity of the face that contains it. Each segment
within the interior of $p$ either separates two faces, and
contributes two to the total complexity of faces within $p$, or
does not separate any face and contributes nothing to the
complexity. Thus in each case the contribution to $p$ and to
the sum of its faces is the same modulo two.\qed
\end{proof}
\begin{lemma}
\label{lem:odd-polygon-is-face} Let $p$ be an odd polygon in a
drawing of a degree-one graph. Then there exists an odd face
polygon in the same drawing.
\end{lemma}
\begin{proof}
By Lemma~\ref{lem:polygon-is-sum-of-faces}, the complexity of
$p$ has the same parity as the sum of the complexities of its
faces. Therefore, if $p$ is odd, it has an odd number of odd
faces, and in particular there must be a nonzero number of odd
faces.\qed
\end{proof}
\begin{lemma}
\label{lem:no-switches-from-no-odd-face} Let $D$ be a drawing
of a degree-one graph. Then $D$ has a casing with no switches
if and only if it has no odd face polygon.
\end{lemma}
\begin{proof}
As we have seen, $D$ has a casing with no switches if and only
if the edge crossing graph is bipartite. This graph is
bipartite if and only if it has no odd cycles, and an odd cycle
in the edge crossing graph corresponds to an odd polygon in
$D$. For, if $C$ is an odd cycle in the edge crossing graph, it
must lie on a polygon $p$ of $D$. Each crossing in $C$
contributes one to the complexity of this polygon. Each edge of
$D$ that crosses $p$ without belonging to $C$ either crosses it
an even number of times (contributing that number of additional
segments to the complexity of $p$) and has both endpoints
inside $p$ or both outside $p$, or it crosses an odd number of
times and has one endpoint inside $p$; thus, it contributes an
even amount to the complexity of $p$. Thus, $p$ must be an odd
polygon. By Lemma~\ref{lem:odd-polygon-is-face}, there is an
odd face polygon in $D$. Conversely, any odd face polygon in
$D$ can be shown to form an odd cycle in the edge crossing
graph.\qed
\end{proof}
%
%%% Don't we want the actual time bound in the theorem below?
\begin{theorem}
{\sc MinTotalSwitches} in the weaving model can be solved in
time $O(qk+q^{5/2}\log^{3/2} k)$, where $k$ denotes the number
of crossings in the input drawing and $q$ denotes the number of
its odd face polygons.
\end{theorem}
\begin{proof}
Let $D$ be the drawing which we wish to case for the minimum
number of switches. By Lemma~\ref{lem:degree-one-equivalent},
we may assume without loss of generality that each vertex of
$D$ has degree one.

We apply a solution technique related to the Chinese Postman
problem, and also to the problem of via minimization in VLSI
design~\cite{CheKajCha-TCS-83}: form an auxiliary graph $G^o$,
and include in $G^o$ a single vertex for each odd face polygon in
$D$. Also include in $G^o$ an edge connecting each pair of
vertices, and label this edge by the number of segments of the
drawing that are crossed in a path connecting the corresponding
two faces in $D$ that crosses as few segments as possible. We
claim that the minimum weight of a perfect matching in $G^o$
equals the minimum total number of switches in any casing of
$D$.

In one direction, we can case $D$ with a number of switches
equal to or better than the weight of the matching, as follows:
for each edge of the matching, insert a small break into each
of the segments in the path corresponding to the edge. The
resulting broken arrangement has no odd face cycles, for the
breaks connect pairs of odd face cycles in $D$ to form larger
even cycles. Therefore, by
Lemma~\ref{lem:no-switches-from-no-odd-face}, we can case the
drawing with the breaks, without any switches. Forming a
drawing of $D$ by reconnecting all the break points adds at
most one switch per break point, so the total number of
switches equals at most the weight of the perfect matching.

In the other direction, suppose that we have a casing of $D$
with a minimum number of switches; we must show that there
exists an equally good matching in $G^o$. To show this,
consider the drawing formed by inserting a small break
in each segment of $D$ having a switch. This eliminates all switches
in the drawing, so by
Lemma~\ref{lem:no-switches-from-no-odd-face}, the modified
drawing has no odd face polygons. Consider any face polygon in
the modified drawing;  by Lemma~\ref{lem:odd-polygon-is-face}
it must include an even number of odd faces in the original
drawing. Thus, the odd faces of $D$ are connected in
groups of evenly many faces in the modified drawing, and within
each such group we can connect the odd faces in pairs by paths
of breaks in the drawing, giving a matching in $G^o$ with total
weight at most equal to the number of switches in~$D$.

The number of vertices of the graph $G^o$ is $O(q)$, where $q$
is the number of odd face polygons in $D$. We can construct
$G^o$ in time $O(qk)$ where $k$ is the number of crossings in
$D$ by using breadth-first search in the arrangement dual to
$D$ to find the distances from each vertex to all other
vertices. A minimum weight perfect matching in a complete
weighted graph with integer weights bounded by $k$ can be found
in time $O(q^{5/2}\log^{3/2} k)$ using the algorithm of Gabow
and Tarjan \cite{gabow:91}. Therefore the time for this
algorithm is $O(qk+q^{5/2}\log^{3/2} k)$.\qed
\end{proof}
%
% The Gabow and Tarjan reference is:
% H. N. Gabow and R. E. Tarjan (1991).
% Faster scaling algorithms for general graph matching problems.
% J. ACM 38: 815--853.
% (Note that its time bound has an alpha factor that goes away for dense graphs.)
%
%For {\sc MinMaxSwitches}, we have the following weaker result.
%%
%\begin{theorem}
%If the edge crossing graph $\GI$ is planar, then we can assure
%at most four switches along any edge in the weaving model.
%\end{theorem}
%%%
%\begin{proof}
%Find an edge $e$ that intersects at most five others, take $e$
%out, and recursively treat the rest. Put $e$ back in such a way
%that it does not increase the switch count of any of the edges
%it intersects. Edge $e$ itself will have at most four switches.\qed
%\end{proof}

\section{Minimizing tunnels}\label{sec:tunnels}

In this section we present three algorithms that solve {\sc
MinMaxTunnels}, {\sc MinMaxTunnelLength}, and {\sc
MaxMinTunnelDistance} in the stacking model. We also present
algorithms for {\sc MinMaxTunnels} and {\sc
MaxMinTunnelDistance} in the weaving model. {\sc
MinMaxTunnelLength} is NP-hard in the weaving model.

\subsection{Stacking model}

In the stacking model, some edge $e$ has to be bottommost. This
immediately gives the number of tunnels of $e$, the total
length of tunnels of $e$, and the shortest distance between two
tunnels of $e$. The idea of the algorithm is to determine for
each edge what its value would be if it were bottommost, and
then choose the edge that is best for the optimization to be
bottommost (smallest value for {\sc MinMaxTunnels} and {\sc
MinMaxTunnelLength}, and largest value for {\sc
MaxMinTunnelDistance}). The other $m-1$ edges are stacked
iteratively above this edge. It is easy to see that such an
approach indeed maximizes the minimum, or minimizes the
maximum. We next give an efficient implementation of the
approach. The idea is to maintain the values of all not yet
selected edges under consecutive selections of bottommost edges
instead of recomputing it.

%%% I used "expected time" in the results because we say
%%% we use Mulmuley's algorithm. If we want deterministic
%%% we need a different reference, like Balaban?
We start by computing the arrangement of edges in $O(m\log
m+k)$ expected time, for instance using Mulmuley's
algorithm~\cite{mul94}. This allows us to determine the value
for all edges in $O(k)$ additional time.

For {\sc MinMaxTunnels} and {\sc MinMaxTunnelLength}, we keep
all edges in a Fibonnacci heap on this value. One selection
involves an {\sc extract-min}, giving an edge $e$, and
traversing $e$ in the arrangement to find all edges it crosses.
For these edges we update the value and perform a {\sc
decrease-key} operation on the Fibonnacci heap. For {\sc
MinMaxTunnels} we decrease the value by one and for {\sc
MinMaxTunnelLength} we decrease by the length of the crossing,
which is ${\mathit{casing width}}/\sin\alpha$, where $\alpha$
is the angle the crossing edges make. For {\sc MinMaxTunnels}
and {\sc MinMaxTunnelLength} this is all that we need. We
perform $m$ {\sc extract-min} and $k$ {\sc decrease-key}
operations. The total traversal time along the edges throughout
the whole algorithm is $O(k)$. Thus, the algorithm runs in
$O(m\log m+k)$ expected time.

For {\sc MaxMinTunnelDistance} we use a Fibonnacci heap that
allows {\sc ex\-tract-max} and {\sc increase-key}. For the
selected edge we again traverse the arrangement to update the
values of the crossing edges. However, we cannot update the
value of an edge in constant time for this optimization. We
maintain a data structure for each edge that maintains the
minimum tunnel distance in $O(\log m)$ time under updates. The
structure is an augmented balanced binary search tree that
stores the edge parts in between consecutive crossings in its
leaves. Each leaf stores the distance between these crossings.
Each internal node is augmented such that it stores the minimum
distance for the subtree in a variable. The root stores the
minimum distance of the edge if it were the bottommost one of
the remaining edges. An update involves merging two adjacent
leaves of the tree and computing the distance between two
crossings. Augmentation allows us to have the new minimum in
the root of the tree in $O(\log m)$ time per update. In total this takes $O(m\log m+ k\log m)$ expected time.
\begin{theorem}
Given a straight-line drawing of a graph with $n$ vertices,
$m=\Omega(n)$ edges, and $k$ edge crossings, we can solve {\sc
MinMaxTunnels} and {\sc MinMaxTunnelLength} in $O(m\log m+k)$
expected time and {\sc MaxMinTunnelDistance} in $O(m\log
m+k\log m)$ expected time in the stacking model.
\end{theorem}

\subsection{Weaving model}

In the weaving model, the polynomial time algorithm for {\sc
MinMaxTunnels} comes from the fact that the problem of
directing an undirected graph, and minimizing the maximum
indegree, can be solved in time quadratic in the number of
edges~\cite{v-mmi-04}. We apply this on the edge crossing graph
of the drawing, and hence we get $O(m^4)$ time. For minimizing
tunnel length per edge, we can show:
\begin{theorem}
{\sc MinMaxTunnelLength} is NP-hard in the weaving model.
\end{theorem}
\begin{proof} The reduction is from {\sc planar 3-sat}, shown NP-hard by Lichtenstein~\cite{l-pfu-82}. The reduction is similar to the one for maximizing minimum visible perimeter length in sets of opaque disks of unit size~\cite{chks-aapsm-06}.
Note that the proof implies that no PTAS exists. The reduction
only uses edges that intersect two or three other edges, so
restricting the number of intersections per edge to be constant
leaves the problem NP-hard. Also, the number of orientations of
edges is constant.

A cased drawing of a set of line segments has property (A) if every line segment has at most two tunnels at crossings with a perpendicular segment, or one tunnel at a crossing with a non-perpendicular segment. Our reduction is such that a {\sc planar 3-sat} instance is satisfiable if and only if a set of line segments has a cased drawing with property (A).

We arrange a set of line segments of equal length, using only four orientations. The slopes are $-4$, $-\frac{1}{4}$, $+\frac{1}{4}$, and $+4$. If two perpendicular line segments cross, then one has tunnel length equal to the width $w$ of the casing at the crossing. If two other line segments cross, then one edge has tunnel length $w/\sin(\gamma)= 2,125\cdot w$ at the crossing, where $\gamma=2\cdot\arctan(\frac{1}{4})$ is the (acute) angle between the line segments. Therefore, a cased drawing with property (A) has tunnel length at most $2,125\cdot w$, whereas a cased drawing that does not satisfy property (A) has an edge that has tunnel length at least $3\cdot w$. This shows the direct relation between property (A) and {\sc MinMaxTunnelLength}, and provides the gap that shows that no PTAS exists.
\begin{figure}[b]
\begin{center}
\vspace{-\baselineskip}
\includegraphics[width=.9\textwidth]{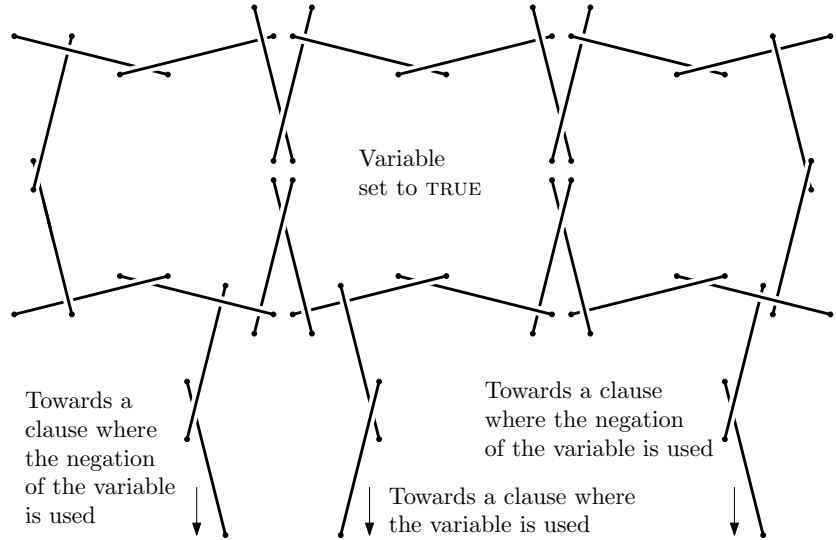}
\end{center}
\vspace{-\baselineskip}
\caption{Boolean variable and the connection of channels.}
\label{f:NPvariable}
\end{figure}

A Boolean variable $x_i$ is modeled by a cycle of crossing line
segments as in Fig.~\ref{f:NPvariable}. Along the cycle,
crossings alternate between perpendicular and
non-perpendicular, and hence it has even length. The variable
satisfies property~(A) iff the cycle has cyclic overlap, which
can be clockwise or counterclockwise. One state is associated
with $x_i=\,${\sc true}, the other is associated with
$x_i=\,${\sc false}. In each state, the line segments of the
cycle alternate in allowing an additional, perpendicular line
segment to have a bridge over the line segment of the cycle. In
the figure, where the cycle is in the {\sc true}-state, the
line segments with slope $+\frac{1}{4}$ and $+4$ allow such an
extra tunnel under a line segment that is not from the cycle.
If the cycle is in the {\sc false}-state, the line segments
with slope $-4$ and $-\frac{1}{4}$ allow the extra tunnel. We
use the line segments of slope $-\frac{1}{4}$ to make
connections and channels to clauses where $\oli{x_i}$ occurs,
and the line segments with slope $+\frac{1}{4}$ for clauses
where $x_i$ occurs. Note that the variable can be made larger
easily to allow more connections, in case the variable occurs
in many clauses.

Channels are formed by line segments that do not cross
perpendicularly. So any line segment of the channel can have
a tunnel at at most one of its two crossings, or else property (A)
is violated. Note that a sequence of crossing line segments with
slopes such as $-4,\,+4,\,+\frac{1}{4},\,-\frac{1}{4}$ gives
a turn in the channel. The exact position of the crossing is
not essential and hence we can easily reach any part of the
plane with a channel, and ending with a line
segment of any orientation.
\begin{figure}[t]
\begin{center}
\includegraphics{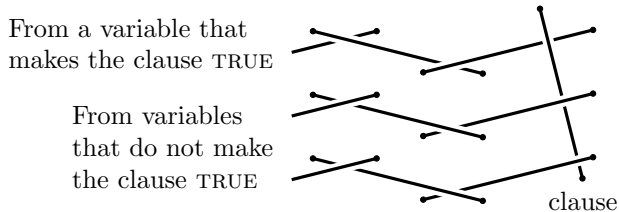}
\end{center}
\vspace{-1.5\baselineskip}
\caption{A clause construction.}
\vspace{-\baselineskip}
\label{f:NPclause}
\end{figure}
A {\sc 3-sat} clause is formed by a single line segment that is
crossed perpendicularly by three other line segments, see
Fig.~\ref{f:NPclause}. Property~(A) holds if the clause
line segment has at most two tunnels. This corresponds
directly to satisfiability of the clause.

With this reduction, testing if property~(A) holds is equivalent
to testing if the {\sc planar 3-sat} instance is satisfiable, and
NP-hardness follows.\qed
\end{proof}
In the remainder of this section we show how to solve {\sc
MaxMinTunnelDistance}. We observe that there are polynomially
many possible values for the smallest tunnel distance, and
perform a binary search on these, using {\sc 2-sat} instances as the
decision tool.

We first compute the arrangement of the $m$ edges
to determine all crossings. Only distances between two---not
necessarily consecutive---crossings along any edge can give the
minimum tunnel distance. One edge crosses at most $m-1$ other
edges, and hence the number of candidate distances, $K$, is
$O(m^3)$. Obviously, $K$ is also $O(k^2)$. From the arrangement
of edges we can determine all of these distances in $O(m\log m
+ K)$ time. We sort them in $O(K \log K)$ time to set up a
binary search. We will show that the decision step takes
$O(m+K)$ time, and hence the whole algorithm takes $O(m\log m +
K\log K)=O((m+K)\log m)$ time.

Let $\delta$ be a value and we wish to decide if we can set the
crossings of edges such that all distances between two tunnels
along any edge is at least $\delta$. For every two edges $e_i$
and $e_j$ that cross and $i<j$, we have a Boolean variable
$x_{ij}$. We associate $x_{ij}$ with {\sc true} if $e_i$ has a
bridge at its crossing with $e_j$, and with {\sc false}
otherwise. Now we traverse the arrangement of edges and
construct a {\sc 2-sat} formula. Let $e_i$, $e_j$, and $e_h$ be three
edges such that the latter two cross $e_i$. If the distance
between the crossings is less than $\delta$, then $e_i$ should
not have the crossings with $e_j$ and $e_h$ as tunnels. Hence,
we make a clause for the {\sc 2-sat} formula as follows
(Fig.~\ref{fig:2sat}): if $i<j$ and $i<h$, then the clause is
$(x_{ij} \vee x_{ih})$; the other three cases ($i>j$ and/or
$i>h$) are similar. The conjunction of all clauses gives a
{\sc 2-sat} formula that is satisfiable if and only if we can set the
crossings such that the minimum tunnel distance is at least
$\delta$. We can construct the whole {\sc 2-sat} instance in $O(m+K)$
time since we have the arrangement, and satisfiability of {\sc 2-sat}
can be determined in linear time~\cite{EveItaSha-SJC-76}.
\begin{figure}[t]
\begin{center}
\includegraphics{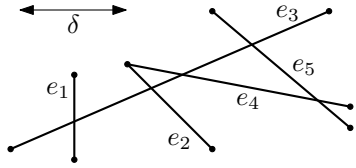}
\end{center}
\vspace{-\baselineskip}
\caption{The {\sc 2-sat}
formula $(\ox_{13}\vee \ox_{23}) \wedge (\ox_{23}\vee
x_{34}) \wedge (\ox_{23} \vee x_{35}) \wedge (x_{34}\vee
x_{35})$.} \label{fig:2sat}
\end{figure}

\begin{theorem}
Given a straight-line drawing of a graph with $n$ vertices and
$m=\Omega(n)$ edges,  we can solve {\sc MaxMinTunnelDistance}
in $O((m+K)\log m)$ expected time in the weaving model, where
$K=O(m^3)$ is the total number of pairs of crossings on the
same edge.
\end{theorem}

\section{Conclusions and Open Problems}

We presented polynomial time algorithms or NP-hardness results for a number of optimization problems that are motivated by cased drawings. Naturally, we would like to establish the difficulty of the {\sc MinMaxSwitches} problem. We would also like to implement our algorithms to visually evaluate the quality of the resulting drawings.

%---------------------------- Bibliography -------------------------------

\bibliographystyle{abbrv}
\bibliography{TB}

%------------------------------ Appendix ---------------------------------

\newpage
\appendix

\section{Removing the restrictions}\label{sec:restrictions}

The restrictions that we impose on the input drawing---no vertex lies on a or very close to an edge, no more than two edges cross in one point, and crossing are not too close---have a significant influence on the problems that we study. For example, if we allow edges to partly overlap with vertices, then it is natural to always draw the vertex on top with a casing around it. In this case, however, there might be no cased drawing in the stacking model, because three edges that necessarily give cyclic overlap can easily be constructed. Testing if cyclic overlap occurs, and hence, if a cased drawing in the stacking model exists can be done by topological sort.

If we remove the restriction that edge crossings are not too close,
we may have three edges that intersect so close that it is not
possible for the edges to have a tunnel at the one crossing and a
bridge at the other. Consequently, three such edges must (locally) be
stacked, where one edge has a single bridge over both other edges,
and both other edges have a single tunnel. Allowing such ``triple
crossings'' makes the problem {\sc MinTotalSwitches} NP-hard.
\begin{theorem}
If triple crossings of edges are allowed, then {\sc MinTotalSwitches}
is NP-hard in both the weaving and the stacking model.
\end{theorem}
\begin{proof}
By reduction from {\sc planar max-2sat} shown NP-hard by~\cite{ghms-apsml-93}. Assume an instance of such a problem which has $n$ Boolean variables and $m$ clauses. The objective is to satisfy as many of the {\sc 2sat} clauses as possible. We reduce the instance to an instance of cased drawing such that the number of switches is the same as the number of unsatisfied clauses in the {\sc planar max-2sat} instance.

\begin{figure}[h]
\begin{center}
\includegraphics[width=.9\textwidth]{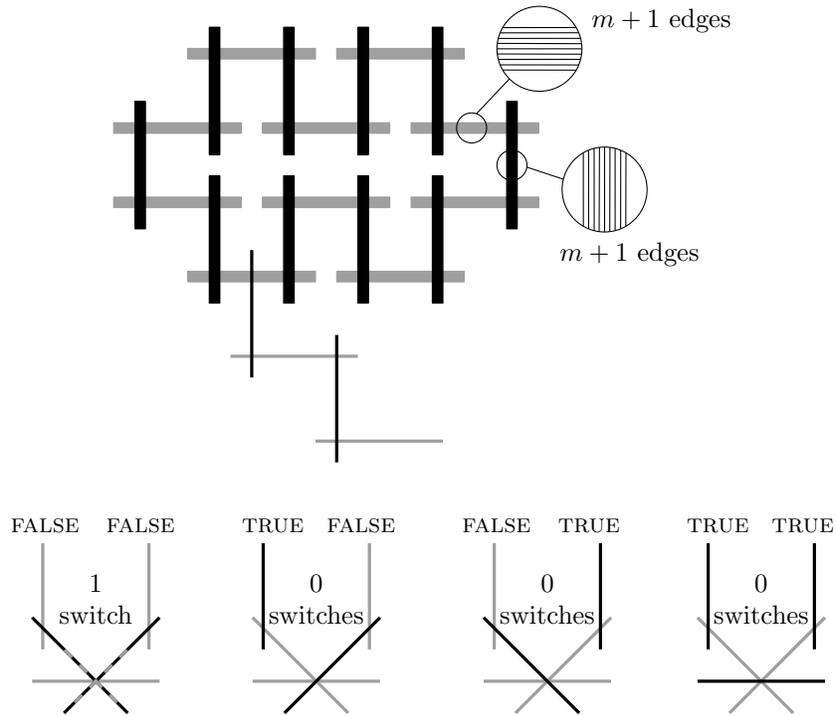}
\end{center}
\caption{A variable, a channel, and four cases of a clause construction.}
\label{fig:max2sat}
\end{figure}

A variable is an even cycle of bundles of edges. Each bundle
consists of $m+1$ parallel edges that are either horizontal or vertical.
A horizontal bundle intersects two vertical bundels, and each
vertical bundle intersects two horizontal bundles, see Fig.~\ref{fig:max2sat}.
Since the variables part of the construction is bipartite, it need not
contain any switches. Either all horizontal edges are on top of the vertical
ones, or vice versa. This corresponds to the {\sc true} and {\sc false}
assignment of the variable, respectively.

We tap off a channel from a variable construction with a single vertical
edge as in Fig.~\ref{fig:max2sat}. The edge crosses all edges of a horizontal
bundle, and to avoid switches, it must be below or above all of them
(depending on whether the vertical edges in the variable are below or above).
The channel itself is a sequence of single edges that leads from the variable
construction to a clause construction.
A channel has no switches if the edges alternate in having
bridges at both crossings and having tunnels at both crossings.
The parity of the number of edges in a channel determines if the channel
arrives at a clause with an edge that already has a tunnel, or with an
edge that already has a bridge. By using horizontal, vertical, and diagonal
edges, we can always get the parity as desired and end with an edge in
any orientation.
So far, no switches are needed yet in a {\sc MinTotalSwitches} cased drawing.

A clause consists of a single edge that has a triple crossing with the two
channels that come from variables that occur in the clause. The clause
edge itself cannot have a switch, because it only has the triple intersection.
To avoid switches on the last edges of the channels, we can put the clause
edge on top if both channel edges already have a tunnel, or we can put the
clause edge at the bottom otherwise. If both channel edges already have a
bridge, then there is no way to avoid a switch because one of the them will
have to go below the other. There will be exactly one switch; in all other
cases there are no switches at the clause. So we let the parity of any chain
be such that the last channel edge be such that it has a bridge if the
variable has a state that gives {\sc false} in the clause.

The maximum number of switches needed is obviously at most $m$, the number of
clauses of the {\sc planar max-2sat} instance. We used bundles of $m+1$
parallel edges in the variables to make sure that a variable cannot have
a mixed state without giving at least $m+1$ switches already. If the bundles
were just single edges, then with only two switches, one part of the channels can
use the {\sc true} state and the other part the {\sc false} state of the variable,
possibly satisfying many more clauses. The use of bundles prevents this
possibility.

Note that the minimum number of switches will be achieved with a stacking
drawing. Hence, the problem is NP-hard in both models.\qed
\end{proof}

\section{Optimal drawing in the weaving model}\label{sec:optweave}

Fig.~\ref{fig:optWeaveNotRealizable} shows that the weaving model is stronger than the realizable model for {\sc MinTotalSwitches}---no cased drawing of this graph in the realizable model can reach the optimum number of 12 switches. The reasoning is quite similar to the one employed for the construction in Fig.~\ref{fig:weavingopt}.
The thickly drawn bundles of $c>12$ parallel edges must be cased as shown (or its mirror image) else there would be at least $c$ switches in each bundle. The 8 vertical and horizontal ``normal'' edges must cross the bundles consistently with the casing of the bundles, and this already leads to the 12 switches as drawn in this figure. Thus, any deviation from the drawing in the casing of the 16 crossings between the normal edges would create additional switches. However, the normal edges as drawn here constitute a perfect $4 \times 4$ weaving which can not be realized (see~\cite{PacPolWel-Algo-93}).
\begin{figure}[b]
  \centering
  \vspace{-\baselineskip}
  \includegraphics[width=.55\textwidth]{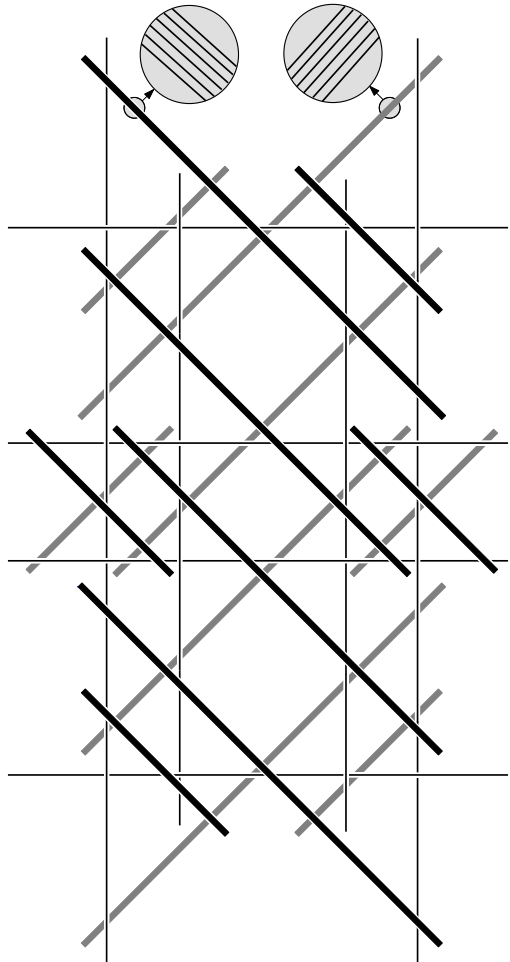}
  \vspace{-\baselineskip}
  \caption{Optimal drawing in the weaving model for {\sc MinTotalSwitches}.}
  \label{fig:optWeaveNotRealizable}
\end{figure}

\end{document}